\documentclass[10pt, doublecolumn]{IEEEtran}
\usepackage{epsfig,latexsym}
\usepackage{float}
\usepackage{indentfirst}
\usepackage{amsmath}
\usepackage{bm}
\usepackage{amssymb}
\usepackage{times}

\usepackage{algorithm}
\usepackage[noend]{algpseudocode}
\usepackage{subfigure}
\usepackage{psfrag}
\usepackage{hyperref}
\usepackage{cite}
\usepackage{lastpage}
\usepackage{fancyhdr}
\usepackage{color}
 \usepackage{amsthm}
\usepackage{bigints}
\usepackage{booktabs}
\usepackage{multirow}
\usepackage{siunitx}

\newtheorem{Lemma}{Lemma}
\newtheorem{lemma}[Lemma]{$\mathbf{Lemma}$}

\newcounter{problem}
\newcounter{save@equation}
\newcounter{save@problem}
\makeatletter
\newenvironment{problem}
{\setcounter{problem}{\value{save@problem}}%
  \setcounter{save@equation}{\value{equation}}%
  \let\c@equation\c@problem
  \subequations
}
{\endsubequations
  \setcounter{save@problem}{\value{equation}}%
  \setcounter{equation}{\value{save@equation}}%
}

\begin{document}
\title{ \huge{Utilizing Imperfect Resolution of Near-Field Beamforming:   \\A Hybrid-NOMA Perspective  }}

\author{\author{ Zhiguo Ding, \IEEEmembership{Fellow, IEEE}  and H. Vincent Poor, \IEEEmembership{Life Fellow, IEEE}   \thanks{ 
  
\vspace{-2em}

  Z. Ding is with  Khalifa University, Abu Dhabi, UAE, and  University of Manchester, Manchester, M1 9BB, UK.    H. V. Poor is  with   Princeton University, Princeton, NJ 08544,
USA.   }
 
 \vspace{-2em}

  }\vspace{-5em}}
 \maketitle

\vspace{-1em}
\begin{abstract}
This letter studies how  the imperfect resolution of near-field beamforming, the key feature of near-field communications, can be used to improve  the throughput and connectivity of wireless networks. In particular, a hybrid non-orthogonal multiple access (NOMA) transmission strategy is developed to use  preconfigured near-field beams for serving additional users.  An energy consumption minimization problem is first formulated and then solved by using different successive interference cancellation strategies. Both analytical and simulation results are presented to illustrate the impact of the resolution  of near-field beamforming  on the design of hybrid NOMA transmission. 
\end{abstract}\vspace{-0.5em}

\vspace{-1em} 
  
  \section{Introduction}
 Near-field communication has recently received considerable  attention because of its innovative way of utilizing spatial degrees of freedom \cite{devar,10129111, 9738442,Eldar2}. Different from conventional far-field beamforming which relies on the approximated planar channel model, near-field beamforming is built on a more accurate spherical model. The recent studies in \cite{10195974, zhaolin} showed that the use of the spherical model enables beamfocusing, i.e., users with different locations can be distinguished  in the so-called distance-angle domain, whereas   conventional far-field beamforming   relies on the angle domain only. When the users' distances to the base station are proportional to the Rayleigh distance,  the resolution of near-field beamforming, i.e.,  whether users can be distinguished  in the distance-angle domain, has been shown to be poor  \cite{dingreso}. However, if the users are extremely close to the base station,  indeed the resolution of near-field beamforming is almost perfect, as shown in Fig. \ref{fig1a} \cite{ldmadai}. When there are multiple users that  are very close to the base station, Fig. \ref{fig1b} shows that near-field beamforming can still effectively suppress the interference among these users. However, Fig. \ref{fig1b} also reveals an important observation that the magnitude of these users' near-field beams   can be accumulated at a far location, which presents an opportunity  to use these preconfigured near-field beams for supporting additional users.
 
 Motivated by this observation, this letter focuses on the scenario with multiple near-field legacy users, with   the aim of studying  how the   near-field beams preconfigured to the legacy users can be used to serve additional users. The additional users can be served at the same time as the legacy users, as in those existing  strategies based on pure non-orthogonal multiple access (NOMA) in \cite{10129111,dingreso}. Different from these existing works,    a hybrid NOMA strategy is considered in this letter, where the existing pure NOMA strategies and conventional orthogonal multiple access (OMA) belong to this general framework. We note that hybrid NOMA has been applied to the uplink scenario only  \cite{9679390}, whereas the downlink scenario is   the focus of   this letter. An energy consumption minimization problem is first formulated, and then solved based on different successive interference cancellation (SIC) strategies. Both analytical and numerical results are obtained for insightful understandings about   hybrid NOMA assisted  near-field communications. In particular, hybrid NOMA is shown to be   preferable to pure NOMA in terms of energy consumption. In addition, when the resolution of near-field beamforming is perfect, it is preferable to use all the preconfigured beams to support additional users. However, if the resolution of near-field beamforming is poor, it is crucial to carry out beam selection for energy minimization.

   \begin{figure}[t] \vspace{-0em}
\begin{center}
\subfigure[A Single-User Case]{\label{fig1a}\includegraphics[width=0.24\textwidth]{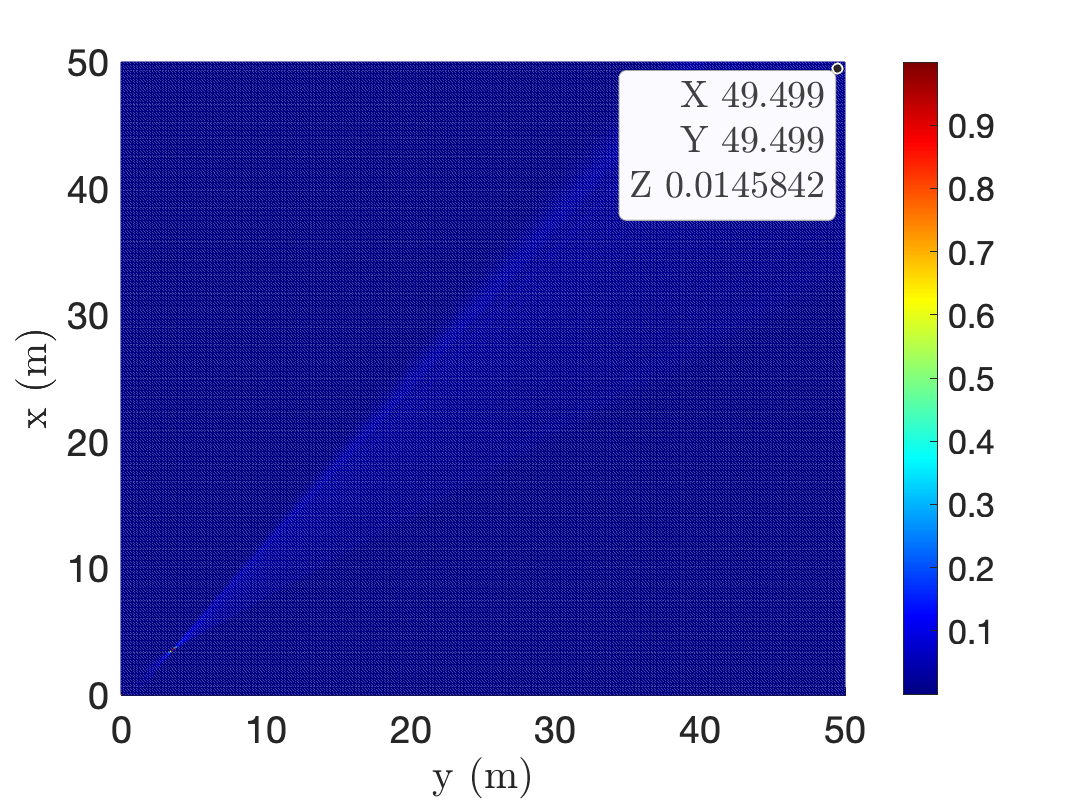}} 
\subfigure[A Three-User Case]{\label{fig1b}\includegraphics[width=0.24\textwidth]{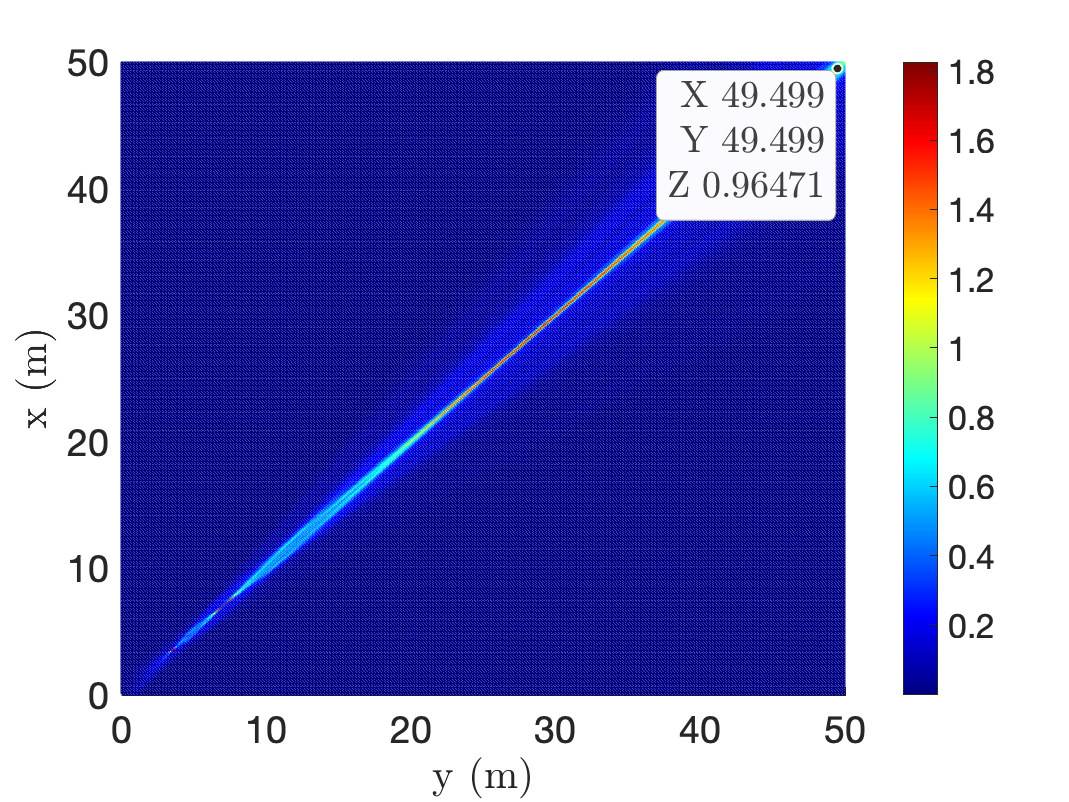}} \vspace{-2em}
\end{center}
\caption{ Illustration of the  beamfocusing resolution ($\Delta$), with a $513$-antenna uniform linear array, where the carrier frequency is $f_c=28$ GHz,   the antenna spacing is $d=\frac{\lambda}{2}$ and $\lambda $ denotes the wavelength. For Fig. \ref{fig1a}, there is a single user located at  $\left(5, \frac{\pi}{4}\right)$, and for Fig. \ref{fig1b}, the users' locations are $\left(5, \frac{\pi}{4}\right)$, $\left(10, \frac{\pi}{4}\right)$, and $\left(40, \frac{\pi}{4}\right)$, respectively, where   polar coordinates are used.       \vspace{-1em} }\label{fig1}\vspace{-1em}
\end{figure}


\vspace{-1em}
\section{System Model}
Consider a downlink legacy network, where there exist  $M$ near-field   users, which are denoted by ${\rm U}_{m}$, $1\leq m \leq M$, and served by a  base station   equipped with  an $N$-antenna uniform linear array (ULA). The aim of the letter is to use the beams preconfigured to the legacy users to serve an additional user, denoted by ${\rm U}_0$.  Each user is assumed to have a single antenna.  Consider that the ULA is placed at the center of a  $2$-dimensional  plane. Denote the locations of  ${\rm U}_m$,   the center of the ULA, and   the $n$-th element of the ULA by  ${\boldsymbol \psi}_m^{\rm U}$,  ${\boldsymbol \psi}_0$, and ${\boldsymbol \psi}_n$, respectively, where Cartesian coordinates are used.  It is assumed that the legacy users are very close to the base station in order to facilitate the implementation of beamfocusing, whereas ${\rm U}_0$ is located far away from the base station.

\vspace{-1em}
\subsection{OMA Transmission}
 To avoid the interference between the two types of users, a straightforward OMA approach  can be adopted by combining space division multiple access (SDMA) with time-division multiple access (TDMA).  
In particular,  during the first $M$ time slots, the base station broadcasts $ \sum^{M}_{m=1}\mathbf{w}_ms_m$, and the legacy user ${\rm U}_m$ receives the following observation:
 \begin{align}
 y_m = \mathbf{h}_m^H \sum^{M}_{m=1}\sqrt{P}\mathbf{w}_ms_m+n_m,
 \end{align} 
where  $s_m$, $\mathbf{w}_m$,  $P$ denote ${\rm U}_m$'s signal, beamforming vector and transmit power, respectively,   $n_m$ denotes the white noise with its power denoted by $\sigma^2$, ${\rm U}_m$'s spherical channel vector is denoted by  $
 \mathbf{h}_m = \sqrt{N}\alpha_m \mathbf{b}\left({\boldsymbol \psi}^{\rm U}_m\right)$, 
  $\mathbf{b}\left({\boldsymbol \psi} \right)=\frac{1}{\sqrt{N}}\begin{bmatrix} 
 e^{-j\frac{2\pi }{\lambda}\left| {\boldsymbol \psi}  -{\boldsymbol \psi}_1\right|} &\cdots &  e^{-j\frac{2\pi }{\lambda}\left| {\boldsymbol \psi}  -{\boldsymbol \psi}_N\right|}
 \end{bmatrix}^T$, $\alpha_m = \frac{\lambda}{4\pi   \left| {\boldsymbol \psi}^{\rm U}_m -{\boldsymbol \psi}_0\right|}$, and $\lambda$ denotes the wavelength \cite{devar,9738442,Eldar2}.   
 Therefore,   ${\rm U}_m$'s data rate in OMA can be expressed as: $ R_m^{\rm OMA} = \log\left(1+\frac{P  |\mathbf{h}_m^H\mathbf{w}_m|^2}{P \sum^{M}_{i=1,i\neq m} |\mathbf{h}_m^H\mathbf{w}_i|^2+\sigma^2}\right)$.

During the $(M+1)$-th time slot, only ${\rm U}_0$ is served   with the following data rate:
 \begin{align}
 R_0^{\rm OMA}\left(P_0^{\rm OMA}\right) = \log\left(1+\frac{P_0^{\rm OMA}}{\sigma^2} |\mathbf{h}_0^H\mathbf{w}_0|^2\right),
 \end{align}
 where $\mathbf{h}_0$, $P_0^{\rm OMA}$ and $\mathbf{w}_0$ denote  ${\rm U}_0$'s channel vector, transmit power, and beamforming vector, respectively.

A conventional design of $\mathbf{w}_m$ is to apply the zero-forcing (ZF) approach, which means that $\mathbf{W} = \mathbf{H}\left(\mathbf{H}^H\mathbf{H}\right)^{-1}\mathbf{D}$, where $\mathbf{W}=\begin{bmatrix}\mathbf{w}_1&\cdots &\mathbf{w}_M
\end{bmatrix}$,   $\mathbf{H}=\begin{bmatrix}\mathbf{h}_1&\cdots &\mathbf{h}_M
\end{bmatrix}$  and $\mathbf{D}=\text{diag}\left(\text{diag}\left(\left(\mathbf{H}^H\mathbf{H}\right)^{-1}\right)\right)^{-\frac{1}{2}}$.
If the resolution of near-field beamforming is perfect, i.e., $\Delta\triangleq \left|\mathbf{b}\left({\boldsymbol \psi}^{\rm U}_m \right)^H\mathbf{b}\left({\boldsymbol \psi}^{\rm U}_n \right)\right|^2\rightarrow 0$ for $m\neq n$,  the beamfocusing based beamforming can also be used, i.e., $\mathbf{w}_m=\mathbf{b}(\boldsymbol \phi^{\rm U}_{m})$.  For  ${\rm U}_0$, $\mathbf{w}_0$ can be simply be  $\mathbf{w}_0=\mathbf{b}(\boldsymbol \phi^{\rm U}_{0})$.  We note that ${\rm U}_0$ might be a far-field user, where $\mathbf{w}_0=\mathbf{b}(\boldsymbol \phi^{\rm U}_{0})$ is still applicable  since the far-field planar channel model is an approximation of the  more accurate spherical model.

  \vspace{-1em}  
 \subsection{Hybrid NOMA Transmission }
The key idea of hybrid NOMA is to   use those   beams preconfigured to the legacy users, $\mathbf{w}_m$, as bandwidth resources to serve ${\rm U}_0$.  In particular, during the first $M$ time slots, the base station sends the following signal 
  \begin{align}
 x^{\rm NOMA}=\sum^{M}_{m=1}\mathbf{w}_m(\sqrt{P}s_m+f_m\sqrt{P_m} s_0),
 \end{align}
 where $P_m$ is ${\rm U}_0$'s transmit power allocated on $\mathbf{w}_m$, and  $f_m$ is the corresponding complex-valued coefficient whose magnitude is one.  Because   ${\rm U}_0$ suffers large path losses    and hence has poor channel conditions, ${\rm U}_0$ is to directly decode its own signal. 
  However, the legacy users,   ${\rm U}_m$, can have two decoding approaches, as illustrated in the following.  
 
 \subsubsection{Approach I}
Each legacy user decodes its own signal directly, which means that its data rate is given by
 \begin{align}
R_{1,m}^{\rm I}=  \log\left(1+\frac{P|\mathbf{h}_m^H\mathbf{w}_m|^2
}{I_m+ \sigma^2}\right) ,
 \end{align}
 where $I_m=P \sum^{M}_{i=1,i\neq m} |\mathbf{h}_m^H\mathbf{w}_i|^2+\left|\mathbf{h}_m^H\sum^{M}_{k=1}\mathbf{w}_k f_k\sqrt{P_k}\right|^2$.  
 To avoid any performance degradation at the legacy users, ${\rm U}_0$'s transmit power needs to be capped; such that $R_{1,m}^{\rm I}\geq R_t$, $1\leq m 
\leq M$, where $R_t$ denotes the legacy user's target data rate.  With this approach,  ${\rm U}_0$'s data rate during the first $M$ time slots is given by
 \begin{align}
R_{1,0}^{\rm I}=&  \log\left(1+\frac{\left|\mathbf{h}_0^H\sum^{M}_{k=1}\mathbf{w}_k f_k\sqrt{P_k}\right|^2}{P \sum^{M}_{i=1} |\mathbf{h}_0^H\mathbf{w}_i|^2+ \sigma^2}\right) .
 \end{align}
 
 \subsubsection{Approach II}
Due to their strong channel conditions,  the legacy users can first decode ${\rm U}_0$'s signal, before decoding their own signals. In particular,   the data rate for ${\rm U}_m$ to decode ${\rm U}_0$'s signal during the first $M$  time slots is capped  as follows:
 \begin{align}\label{r1m}
R_{1,m}^{\rm II}=  \log\left(1+\frac{\left|\mathbf{h}_m^H\sum^{M}_{k=1}\mathbf{w}_k f_k\sqrt{P_k}\right|^2}{P \sum^{M}_{i=1} |\mathbf{h}_m^H\mathbf{w}_i|^2+ \sigma^2}\right) ,
 \end{align}
for $1\leq m \leq M$. After ${\rm U}_0$'s signal is decoded and subtracted, ${\rm U}_m$ can decode its own signal with the same data rate as the one in OMA, i.e., $R_m^{\rm OMA}$.  Because of the constraints shown in \eqref{r1m}, ${\rm U}_0$'s data rate during the first $M$ time slots is given by
\begin{align}
R_{1,0}^{\rm II}=&\min\left\{
R_{1,0}^{\rm I}, R_{1,m}^{\rm II}, 1\leq m \leq M
\right\}. 
\end{align}

\section{Resource Allocation for Hybrid-NOMA}
With hybrid-NOMA, ${\rm U}_0$ has   access to the first $M$ time slots, which is useful for ${\rm U}_0$ to increase its data rate, but potentially increases its energy consumption. Therefore,  this letter focuses on the following energy minimization problem:
  \begin{problem}\label{pb:1} 
  \begin{alignat}{2}
\underset{P_m\geq0,f_m  }{\rm{min}} &\quad    MT \sum^{M}_{m=1} P_m  +TP_0\label{1tst:1}
\\ s.t. &\quad MT R_{1,0}^{k} +T R_0^{\rm OMA}\left(P_0 \right)   \geq T R_0,  \label{1tst:2} 
\\  &\quad \left(R_{1,m}^{\rm I}     - R_t \right)\delta \geq 0, 1\leq m \leq M \label{1tst:3} 
  \end{alignat}
\end{problem}  
for $k\in\{{\rm I, II}\}$, where $\delta=1$ for Approach I, $\delta=0$ for Approach II,  $T$ denotes the duration of each time slot, $P_0$ denotes ${\rm U}_0$'s transmit power in the $(M+1)$-th time slot, and $R_{0}$ denotes ${\rm U}_0$'s target data rate. Depending on which of the two approaches is adopted, problem \ref{pb:1} can be solved differently, as shown in the following.

 {\it Remark 1:} By setting $P_m=0$, $1\leq m \leq M$, hybrid NOMA becomes pure OMA. By setting $P_0=0$, hybrid NOMA becomes pure NOMA. Therefore,   hybrid NOMA is a general framework with pure NOMA and OMA as its special cases.  
 \vspace{-1em}
\subsection{Approach I}
Note that  for beamfocusing,  $\mathbf{w}_m^H\mathbf{w}_n\approx 0$, for $m\neq n$, if the resolution of near-field beamforming is perfect, and $\mathbf{w}_m^H\mathbf{w}_n= 0$ if ZF beamforming is used.  Therefore, $R_{1,m}^{\rm I}$ can be simplified as follows: 
  \begin{align}
R_{1,m}^{\rm I}\approx &  \log\left(1+\frac{P|\mathbf{h}_m^H\mathbf{w}_m|^2
}{ \left|\mathbf{h}_m^H \mathbf{w}_m f_m\sqrt{P_m}\right|^2+ \sigma^2}\right) \\\nonumber
=&
  \log\left(1+\frac{Ph_m
}{ h_m P_m+ \sigma^2}\right) ,
 \end{align}
 where $h_m=\left|\mathbf{h}_m^H \mathbf{w}_m \right|^2$. Because $R_{1,m}^{\rm I}$ is not a function of $f_m$, $f_m$ can be set as follows: $f_m=\frac{\mathbf{h}_0^H \mathbf{w}_m}{|\mathbf{h}_0^H \mathbf{w}_m|}$, which means that 
 ${\rm U}_0$'s data rate can be expressed as follows:
  \begin{align}
R_{1,0}^{\rm I}=&    \log\left(1+a_0\left|\sum^{M}_{k=1} g_k\sqrt{P_k}\right|^2 \right) ,
 \end{align}
 where $g_m=|\mathbf{h}_0^H\mathbf{w}_m|$, $0\leq m \leq M$, and $a_m=\frac{1}{P \sum^{M}_{i=1} |\mathbf{h}_m^H\mathbf{w}_i|^2+ \sigma^2}$, $0\leq m \leq M$. 
 
 By using the simplified expression for the data rates, problem \ref{pb:1} can be simplified as follows:  
  \begin{problem}\label{pb:2} 
  \begin{alignat}{2}
\underset{P_m\geq0  }{\rm{min}} &\text{ }   M \sum^{M}_{m=1} P_m  +P_0   \label{2tst:1}
\\ s.t. &\text{ }   M\log\left(1+a_0\left|\sum^{M}_{k=1} g_k\sqrt{P_k}\right|^2 \right) +   \log\left(1+b_0P_0  \right) \geq R_0  \label{2tst:2} 
\\  &\text{ }   P_m\leq \frac{P}{e^{R_t}-1} -\frac{\sigma^2}{h_m}, 1\leq m\leq M \label{2tst:3} ,
  \end{alignat}
\end{problem}  
where $b_0=\frac{g_0}{\sigma^2}$. 

The constraint function in \eqref{2tst:2} is concave, as illustrated in the following. The left-hand side of \eqref{2tst:2} can be expressed as follows: $f(h(\mathbf{p}), P_0)$, where $\mathbf{p}=\begin{bmatrix}
g_1^2P_1 &\cdots &g_M^2P_M
\end{bmatrix}$, $f(x,P_0) = M \log(1+a_0x)+\log(1+b_0P_0)$, and $h(\mathbf{x}) =\left(\sum^{M}_{k=1}x_k^{\frac{1}{2}}\right)^2 $. Note that $h(\mathbf{x})$ is simply the $\ell_p$-norm of   $\mathbf{x}$ with $p=\frac{1}{2}$. Recall that for   $\mathbf{x}\geq 0$, its $\ell_p$-norm is a concave function, if $p<1$  \cite{Boyd}. Therefore, $f(h(\mathbf{p}), P_0)$ is the composition of two concave functions, and hence the constraint in \eqref{2tst:2} is  concave. Therefore, problem \ref{pb:2} is a concave optimization problem, since 
both the objective function in \eqref{2tst:1} and the constraint in \eqref{2tst:3} are  affine functions.  
 
{\it Remark 2}: Note that $P_m$ is non-negative, and hence the right-hand side of the constrain in \eqref{2tst:3} needs to be non-negative,  i.e.,
\begin{align}\label{remark1}
 \frac{P}{e^{R_t}-1} \geq \frac{\sigma^2}{h_m}\Longrightarrow   
 \log\left(1+\frac{Ph_m}{\sigma^2} \right)\geq R_t.
\end{align}
 which  means that  not all the beams, $\mathbf{w}_m$, are useful to ${\rm U}_0$. In particular,   beam selection   needs to be carried out to ensure that only the qualified beams, i.e., $  h_m \geq  \sigma^2\frac{e^{R_t}-1}{P} $, are used to serve  ${\rm U}_0$. 

\vspace{-1em}
\subsection{Approach II}
Again by applying the orthogonality between the preconfigured beams,    ${\rm U}_m$'s   data rate expression,  $1\leq m \leq M$, can be simplified as follows:
 \begin{align}
R_{1,m}^{\rm II}\approx &  \log\left(1+a_m  \left|\mathbf{h}_m^H \mathbf{w}_m f_m\sqrt{P_m}\right|^2 \right) \\\nonumber
=& \log\left(1+a_mh_m P_m \right) .
 \end{align} 
 Similar to Approach I,  $R_{1,m}^{\rm II}$ is also not a function of $f_m$,  and hence the choice of $f_m=\frac{\mathbf{h}_0^H \mathbf{w}_m}{|\mathbf{h}_0^H \mathbf{w}_m|}$ is still applicable,  which means that problem \ref{pb:1} can be simplified as follows:

  \begin{problem}\label{pb:3} 
  \begin{alignat}{2}
\underset{P_m\geq 0  }{\rm{min}} &\quad     M \sum^{M}_{m=1} P_m  +P_0  \label{3tst:1}
\\ \label{3tst:2}s.t. &\quad M \min \left\{ \log\left(1+a_0\left|\sum^{M}_{k=1} g_k\sqrt{P_k}\right|^2 \right) \right. \\\nonumber &\quad\quad\quad\quad \left., \log\left(1+a_m h_m P_m \right)  , 1\leq m \leq M\right\}\\  &\quad\quad\quad\quad+ \log\left(1+b_0P_0\right)
 \geq  R_0,  \nonumber
  \end{alignat}
\end{problem}  
 Problem \ref{pb:3} can   be further recast   as follows:
    \begin{problem}\label{pb:4} 
  \begin{alignat}{2}\label{4tst:1}
\underset{P_m \geq0 }{\rm{min}} &\text{ }     M \sum^{M}_{m=1} P_m  +P_0 
\\ s.t. &\text{ }  M  \log\left(1+a_m h_m P_m \right)  + \log\left(1+b_0P_0  \right)\label{4tst:2}
 \geq R_0, \\\nonumber &\quad\quad\quad\quad\quad\quad\quad\quad\quad\quad\quad\quad\quad\quad\quad 1\leq m \leq M,  
 \\  &\label{4tst:3}
 M\log\left(1+a_0\left|\sum^{M}_{k=1} g_k\sqrt{P_k}\right|^2 \right)  + \log\left(1+b_0P_0  \right)
 \geq R_0.
  \end{alignat}
\end{problem}  
Similar to problem \ref{pb:2}, problem \ref{pb:4} is also a concave optimization problem, and hence can be efficiently solved. 

{\it Remark 3:} Similar to Approach I, beam selection could be also important to Approach II, as explained in the following. By inviting more beams to help ${\rm U}_0$, ${\rm U}_0$'s effective channel gain, $\left|\sum^{M}_{k=1} g_k\sqrt{P_k}\right|^2$ in \eqref{4tst:3}, is improved; however, the number of constraints in \eqref{4tst:2} is also increased, which can reduce the search space for the power allocation coefficients and hence cause performance degradation. One effective way for beam selection  is to first order the beams according to $h_m$, and select the $M_x$ best beams, $M_x\leq M$.

{\it Remark 4:} We note that the use of  off-shelf optimization solvers can yield the optimal solution of problem  \eqref{pb:4}, but without insightful understandings to the properties of the solutions, which motivates the following special-case studies.

\begin{lemma}\label{lemma1}
Assuming that the preconfigured beams are strictly orthogonal,    the high-SNR approximations of the optimal choices of   $P_0$ and $P_m$, $1\leq m\leq M_x$,  are given by
\begin{align}
P_1^*\approx \cdots \approx P_{M_x}^*\approx \frac{\lambda}{M_x} -\frac{1}{c}, \quad P_0^* \approx \lambda  -\frac{1}{b_0},
\end{align}
for $\frac{P}{\sigma^2}\rightarrow \infty$, if $R_0\geq \frac{M_xb_0}{c}$, otherwise $P_1^*=\cdots =P_{M_x}^*= 0$, and $ P_0^* =\frac{e^{R_0}-1}{b_0}$, 
where the   Lagrange multiplier is given by $\lambda^* 
 =\left(\frac{ M_x^{M}e^{R_0}}{c^{M}b_0}\right)^{\frac{1}{M+1}}$, and $c=\min\left\{\frac{1}{P}, a_0  \left|  \sum^{M_x}_{k=1} g_k\right|^2\right\}$. 
\end{lemma}
\begin{proof}

With the orthogonality assumption made in the lemma, $a_mh_m $ can be approximated at high SNR as follows: 
\begin{align}\label{amhm}
a_mh_m = \frac{\left|\mathbf{h}_m^H \mathbf{w}_m \right|^2}{P \sum^{M}_{i=1} |\mathbf{h}_m^H\mathbf{w}_i|^2+ \sigma^2}\approx \frac{1}{P},
\end{align}
where the orthogonality assumption ensures that $\mathbf{h}_m^H\mathbf{w}_i=0$ for $m\neq i$.  By using \eqref{amhm}, the multiple constraints shown in \eqref{4tst:2} can be merged together, which means that   $P_1=\cdots =P_{M_x}\triangleq P_t$ at high SNR. Therefore, problem \ref{pb:4} can be   recast as follows:
 \begin{problem}\label{pb:7} 
  \begin{alignat}{2}
\underset{P_m \geq0 }{\rm{min}} &\quad    MM_x   P_t +P_0 
\\ s.t. &\quad M  \log\left(1+ \frac{ P_t}{P} \right)  + \log\left(1+b_0P_0  \right)
 \geq R_0,   
 \\\nonumber &
 M\log\left(1+a_0 P_t \left|  \sum^{M_x}_{k=1} g_k\right|^2 \right)  + \log\left(1+b_0P_0  \right)
 \geq R_0,  
  \end{alignat}
\end{problem}  
which can be further simplified  as follows: 
 \begin{problem}\label{pb:8} 
  \begin{alignat}{2}
\underset{P_m \geq0 }{\rm{min}} &\quad    MM_x   P_t +P_0 
\\ s.t. &\quad  
 M\log\left(1+ c P_t   \right)  + \log\left(1+b_0P_0  \right)
 \geq R_0.
  \end{alignat}
\end{problem}

Problem \ref{pb:8} is a concave optimization problem, and its Lagrange can be expressed as follows:
\begin{align}
&L (P_0,P_t,\lambda)=  MM_x   P_t +P_0 \\\nonumber &+\lambda\left(
R_0- M\log\left(1+ c P_t   \right)  -\log\left(1+b_0P_0  \right)
\right),
\end{align}
where $\lambda$ is a Lagrange multiplier  \cite{Boyd}. 
By applying the Karush–Kuhn–Tucker (KKT)  conditions, the lemma can be approved.\end{proof}
 
 {\it Remark 5:} An immediate  conclusion from Lemma \ref{lemma1} is that ${\rm U}_0$'s transmit powers during the first $M+1$ time slots can be reduced by scheduling more beams, i.e., increasing $M_x$ reduces $P_m$, as illustrated in the following.  According Lemma \ref{lemma1}, the optimal value for $P_m$ can be approximated as follows:
 \begin{align}
 P_m^* \approx &\frac{1}{M_x}\left(\frac{ M_x^{M}e^{R_0}}{c^{M}b_0}\right)^{\frac{1}{M+1}} -\frac{1}{c}  \\\nonumber =&
   \left(\frac{  e^{R_0}}{M_xc^{M}b_0}\right)^{\frac{1}{M+1}} -\frac{1}{c}   \approx \frac{1}{c}\left(   \left(\frac{  e^{R_0}}{M_x b_0}\right)^{\frac{1}{M+1}}-1\right) ,
 \end{align}
 where the approximation is obtained by assuming  large $M$. By increasing $M_x$, $c$ is also increasing, which means that $P_m^*$ is reduced. However, it is important to point out that   the overall energy consumption may not be reduced by increasing $M_x$, as shown in the next section. 
 
 For the special case with $M=M_x=1$, an exact expression of power allocation can be obtained without the high-SNR approximation, which leads to more insightful understandings. 
\begin{lemma}\label{lemma2}
For the special case with $M=M_x=1$, the following conclusions can be made:
\begin{itemize}
\item The condition for hybrid NOMA to outperform pure OMA is $  R_0 \geq   \frac{b_0}{\beta}$;
\item The pure NOMA solution, $P_0=0$, causes more energy consumption than hybrid NOMA. 

\end{itemize}
where $\beta=\min \{a_1 h_1,a_0  g_1^2  \}$. 
\end{lemma}
\begin{proof} 
For the two-user case, problem \ref{pb:4} can be simplified as follows:
  \begin{problem}\label{pb:5} 
  \begin{alignat}{2}
\underset{P_m \geq0 }{\rm{min}} &\quad    P_1  +P_{0} 
\\ s.t. &\quad   \log\left(1+a_1 h_1 P_1 \right)  + \log\left(1+b_0P_0  \right)
 \geq R_0
 \\\nonumber &\quad 
 \log\left(1+a_0  g_1^2  {P_1} \right)  + \log\left(1+b_0P_0  \right)
 \geq R_0.
  \end{alignat}
\end{problem}  
By introducing $\beta$, the two constraints can be combined, and problem \ref{pb:5}   can be further simplified as follows:
  \begin{problem}\label{pb:6} 
  \begin{alignat}{2}
\underset{P_m \geq0 }{\rm{min}} &\quad    P_1  +P_{0} 
\\ s.t. &\quad   \log\left(1+ \beta P_1 \right)  + \log\left(1+b_0P_0  \right)
 \geq R_0. 
  \end{alignat}
\end{problem}  
It is straightforward to verify that problem \ref{pb:6} is a concave optimization problem, and the use of  the  KKT conditions leads to the following closed-form expressions of $P_1$ and $P_0$ \cite{Boyd}:
\begin{align}\label{optimal1}
P_1^*=\sqrt{\frac{e^{R_0}}{\beta b_0}}-\frac{1}{\beta}, \quad 
P_0^*=\sqrt{\frac{e^{R_0}}{\beta b_0}}-\frac{1}{b_0},
\end{align}
if $\sqrt{\frac{e^{R_0}}{\beta b_0}}-\frac{1}{\beta}\geq 0$, otherwise pure OMA is to be adopted. 
Establishing the condition for hybrid NOMA to outperform OMA is equivalent to  identify the condition for $P_1^*> 0$. By using the closed-form expression of $P_1^*$ shown in \eqref{optimal1}, the condition can be obtained, as shown in the lemma. 

The second conclusion in the lemma can be proved  by contradiction.  In particular, the difference between pure NOMA and hybrid NOMA means that, if  pure NOMA is adopted, the following condition must be satisfied: 
\begin{align} 
P_0^*=\sqrt{\frac{e^{R_0}}{\beta b_0}}-\frac{1}{b_0}\leq 0,
\end{align}
which is equivalent to the following:
\begin{align} \label{contra1}
 \frac{\beta}{b_0} \geq e^{R_0} \geq 1 \Longrightarrow \beta\geq b_0 . 
\end{align}
 However,  $\beta$ is strictly larger than $b_0$, since 
\begin{align}
\beta =& \min \left\{\frac{\left|\mathbf{h}_1^H \mathbf{w}_1 \right|^2}{P |\mathbf{h}_1^H\mathbf{w}_1|^2+ \sigma^2} ,\frac{|\mathbf{h}_0^H\mathbf{w}_1|^2}{P   |\mathbf{h}_0^H\mathbf{w}_1|^2+ \sigma^2}    \right\}\\\nonumber \leq& \frac{|\mathbf{h}_0^H\mathbf{w}_1|^2}{P   |\mathbf{h}_0^H\mathbf{w}_1|^2+ \sigma^2}  < \frac{|\mathbf{h}_0^H\mathbf{w}_1|^2}{  \sigma^2}  \leq \frac{|\mathbf{h}_0^H\mathbf{w}_0|^2}{  \sigma^2}  =b_0,
\end{align}
which contradicts to the conclusion in \eqref{contra1}. Therefore, the second part of the lemma is proved.  
\end{proof}
{\it Remark 6:} Lemma \ref{lemma2} indicates that pure NOMA suffers a performance loss to hybrid NOMA, which is consistent to the   conclusion     previously reported  to the uplink scenario \cite{9679390}. In addition, Lemma \ref{lemma2} shows that the use of hybrid NOMA is particularly useful for the case with large $R_0$. 
  
    \begin{figure}[t]\centering \vspace{-0em}
    \epsfig{file=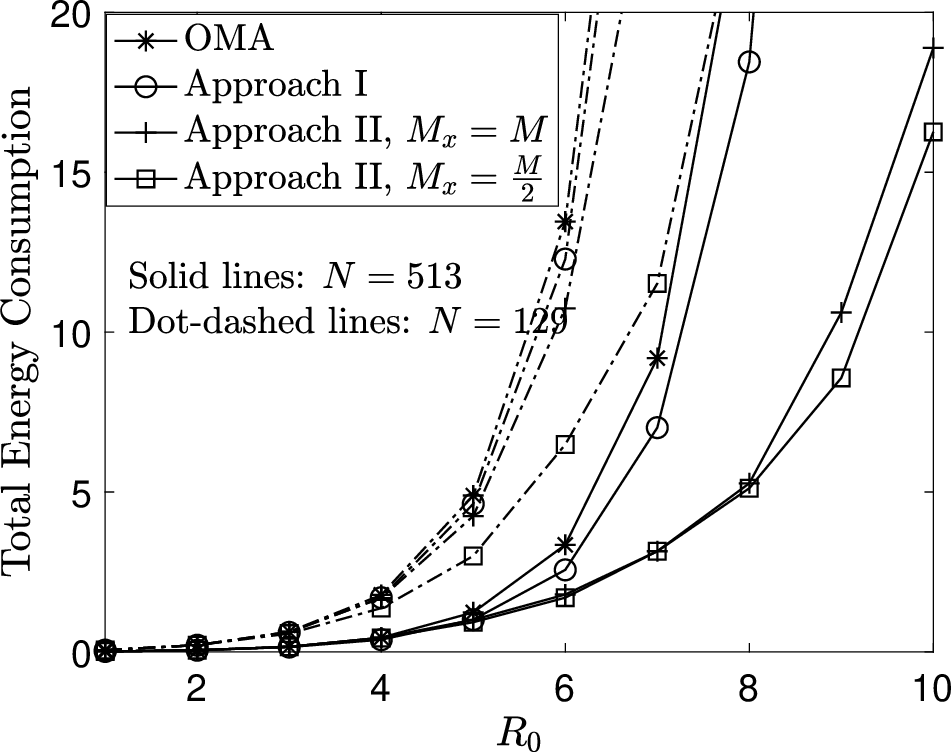, width=0.33\textwidth, clip=}\vspace{-0.5em}
\caption{ Total energy consumption as a function of $R_0$ with randomly located users.  $M=40$, $P=10$ dBm, $\sigma^2=70$ dBm, and $R_t=4$ BPCU.     \vspace{-1em}    }\label{fig1}   \vspace{-1em} 
\end{figure}

\section{Simulation Results}  
In this section, computer simulation results are presented to demonstrate the performance of hybrid NOMA transmission. For all the presented simulation results, $f_c=28$ GHz and the antenna spacing is $d=\frac{\lambda}{2}$.     

In Fig.  \ref{fig1}, the performance of hybrid NOMA is studied by assuming that the users are randomly located.  In particular, the legacy users  are randomly located in a half-ring with its inner radius $5$ m and its outer radius $10$ m. ${\rm U}_0$ is randomly located in a half-ring with its inner radius $150$ m and its outer radius $200$ m. Because the users are randomly located, the users' channel vectors may not be orthogonal, which motivates the use of ZF beamforming. As can be seen from the figure, the use of Approach II can reduce ${\rm U}_0$'s energy consumption significantly, compared to OMA, particularly for large $R_0$, whereas Approach I realizes a performance quite similar to OMA. By increasing the number of antennas of the ULA, the performance gap between hybrid NOMA and OMA can be further enlarged. 

 \begin{figure}[t]\centering \vspace{-0em}
    \epsfig{file=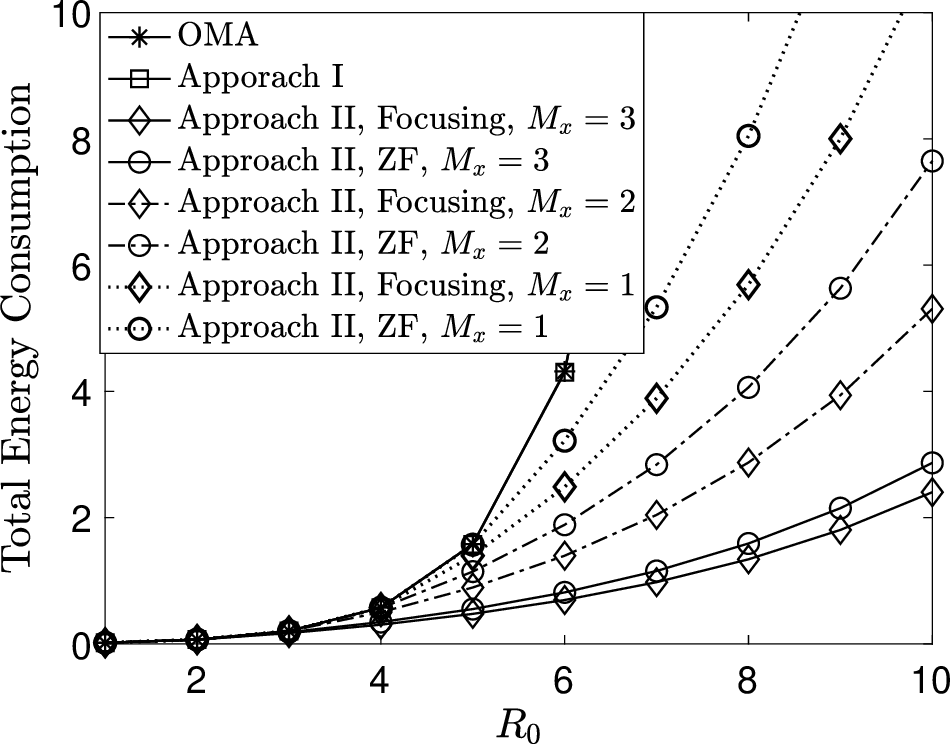, width=0.33\textwidth, clip=}\vspace{-0.5em}
\caption{ The impact of $M_x$ on the energy consumption of hybrid NOMA.   $N=513$, $M=3$, $P=10$ dBm, $\sigma^2=70$ dBm, and $R_t=4$ BPCU.      \vspace{-1em}    }\label{fig2}   \vspace{-1em} 
\end{figure}

In Figs. \ref{fig2} and \ref{fig3}, a deterministic scenario is focused on. In particular,    ${\rm U}_m$, $1\leq m \leq 3$,  are located at $\left(5, \frac{\pi}{4}\right)$, $\left(10, \frac{\pi}{4}\right)$, and $\left(40, \frac{\pi}{4}\right)$, respectively,   ${\rm U}_0$ is located at $\left(200, \frac{\pi}{4}\right)$, where   polar coordinates are used. This scenario represents an important application of near-field communications, such as vehicular networks, where the users are expected to be lined up on a straight street. We note that for this deterministic scenario, it can be easily verified that the channel vectors of ${\rm U}_m$, $1\leq m \leq M$, are almost orthogonal to each other, which means that both beamfocusing and ZF can be used.  Both the figures shows that hybrid NOMA can outperform OMA, particularly with large $R_0$, which confirms Lemma \ref{lemma2}.    In addition, Fig. \ref{fig2} shows that Approach II outperforms Approach I with small $P$ and large $R_t$. However, by increasing $P$ and reducing $R_t$, Approach II can outperform Approach I, as shown in Fig. \ref{fig3}.  Furthermore, Table \ref{table1} also demonstrates the accuracy of the high-SNR approximation    developed in Lemma \ref{lemma1}.

The key feature of the resolution of near-field beamforming causes an interesting observation that increasing $M_x$ introduces  a performance loss   in Fig. \ref{fig1}, but yields a performance gain in Fig. \ref{fig2}, as explained in the following. For the random setup considered in Fig.   \ref{fig2},     the users' channel vectors might be linear dependent, instead of orthogonal, to each other, which means that the use of ZF beamforming can cause the users' channel gains fluctuated.  For example, assume that the users' channel vectors are orthogonal, e.g.,   $\mathbf{H}=\begin{bmatrix}
1 &1\\1&-1
\end{bmatrix}$, which means that  the users' power normalization coefficients are the same, i.e., $\text{diag}\left(\left(\mathbf{H}^H\mathbf{H}\right)^{-1}\right)=\frac{1}{2}\begin{bmatrix}
1  &1
\end{bmatrix}$.  However, if  their channel vectors are not orthogonal, e.g.,  $\mathbf{H}=\begin{bmatrix}
1 &-1\\-1&0
\end{bmatrix}$, the users' power normalization coefficients are no longer the same, i.e.,  $\text{diag}\left(\left(\mathbf{H}^H\mathbf{H}\right)^{-1}\right)= \begin{bmatrix}
1 & 2
\end{bmatrix}$. This fluctuation can make  some users' beams   not as useful as the others', which motivates the use of beam selection, as explained in Remark 3. However, for the case considered in Fig. \ref{fig2}, the resolution of near-field beamforming is perfect, i.e., the users' channel vectors are almost orthogonal to each other, and the fluctuation issue does not exist, which means that  all the beams should be used, i.e., increasing $M_x$ improves the system performance.

   \begin{figure}[t]\centering \vspace{-0em}
    \epsfig{file=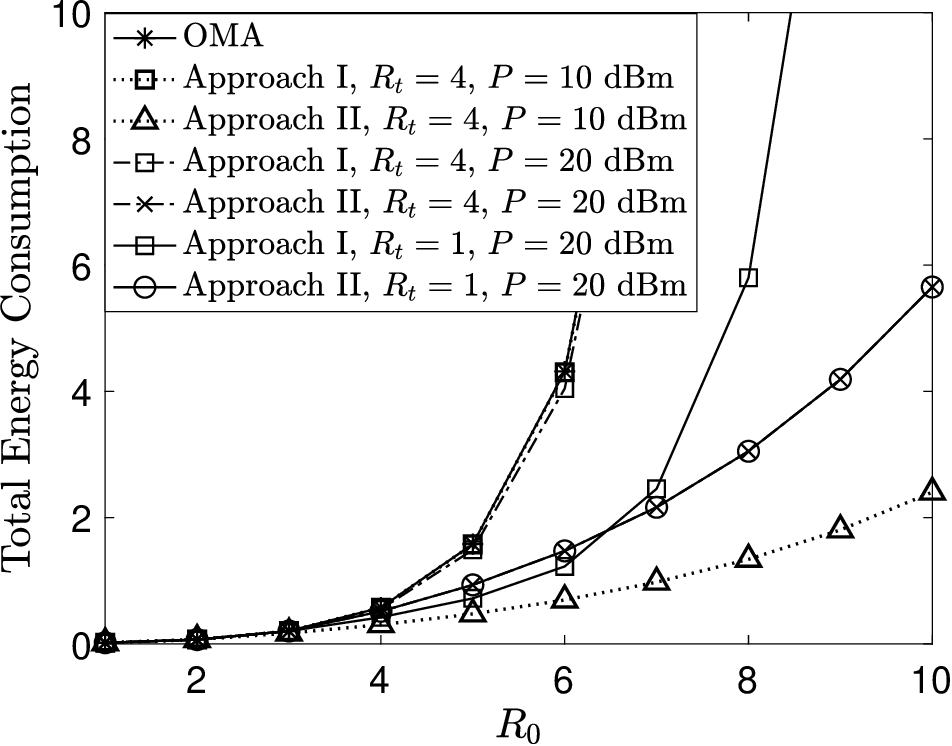, width=0.33\textwidth, clip=}\vspace{-0.5em}
\caption{ Impact of $P$ and $R_t$ on the total energy consumption. $N=513$,  $M=3$, $ M_x=3$ $\sigma^2=70$ dBm.  Beamfocusing is used for beamforming.      \vspace{-1em}    }\label{fig3}   \vspace{-0.01em} 
\end{figure}

\begin{table}\vspace{0em}
  \centering
  \caption{Power Allocation for Hybrid NOMA (in Watts) }\vspace{-1em}
\label{table1}
  \begin{tabular}{|c|c|c|c|c| }
\hline
&  $P_1$ & $P_2$& $P_3$ & $P_4$  \\
    \hline    Optimal ($\sigma^2=-70$ dBm) &0.0684
 	&0.1236 &0.3875	 &0.6647  \\
    \hline Approx.  ($\sigma^2=-70$ dBm)   	
  &0.2118 & 0.2118   & 0.2118  & 0.7338 \\ 
    \hline    Optimal ($\sigma^2=-80$ dBm) &0.0430 
 	&0.0430  &0.0430	 &0.1580  \\
    \hline Approx.  ($\sigma^2=-80$ dBm)   	
 &0.0430 
 	&0.0430  &0.0430	 &0.1579   \\\hline
  \end{tabular}\vspace{-2em}
\end{table}
\vspace{-0.5em}
\section{Conclusions}
 In this letter, a hybrid NOMA transmission strategy was developed to use  the preconfigured near-field beams for serving additional users.  An energy consumption minimization problem was first formulated and then solved by using different successive interference cancellation strategies. Both analytical and simulation results were presented to illustrate the impact of the resolution  of near-field beamforming on the design of hybrid NOMA transmission.
\vspace{-0.5em}

\bibliographystyle{IEEEtran}
\bibliography{IEEEfull,trasfer}
  \end{document}